\documentclass[11pt]{article}

\usepackage{amssymb,amsmath,amsthm}
\usepackage{graphicx}

	\newtheorem{thm}{Theorem}             
	\newtheorem{lem}[thm]{Lemma}
  \DeclareMathOperator{\vis}{\mathrm{visible}}
  
  \DeclareMathOperator{\conv}{\mathrm{conv}}
  \DeclareMathOperator{\per}{\mathrm{per}}
  \newcommand{\eps}{\varepsilon}

\begin{document}

\title{The Visible Perimeter of an Arrangement of Disks\footnote{A preliminary version of this paper appeared in \emph{Graph Drawing 2012} (LNCS 7704, pp. 364--375, 2013).}}

\author{Gabriel Nivasch\footnote{Ariel University, Ariel, Israel. Work done when the author was at EPFL, Lausanne, Switzerland.}\\\texttt{gabrieln@ariel.ac.il} \and J\'anos Pach\footnote{EPFL, Lausanne, Switzerland and R\'enyi Institute, Budapest, Hungary. Work supported by Hungarian Science Foundation EuroGIGA Grant OTKA NN 102029, by Swiss National Science Foundation Grants 200020-144531 and 20021-137574, and by NSF Grant CCF-08-30272.}\\\texttt{pach@cims.nyu.edu}  \and G\'abor Tardos\footnote{R\'enyi Institute, Budapest, Hungary. Work supported by an NSERC grant and by OTKA grants T-046234, AT048826 and NK-62321.}\\\texttt{tardos@renyi.hu}}

\maketitle

\begin{abstract}
Given a collection of $n$ opaque unit disks in the plane, we want to find a stacking order for them that maximizes their {\em visible perimeter}, the total length of
all pieces of their boundaries visible from above.  We prove that if the centers of the disks form a {\em dense} point set,
{\em i.e.}, the ratio of their maximum to their minimum distance
is $O(n^{1/2})$, then there is a stacking order for which the visible
perimeter is $\Omega(n^{2/3})$. We also show that
this bound cannot be improved in the case of a sufficiently small $n^{1/2}\times n^{1/2}$ uniform grid. On the other hand, if the
set of centers is dense and the maximum distance between them is
small, then the visible perimeter is $O(n^{3/4})$ with respect to any stacking order. This latter bound cannot be improved either.

Finally, we address the case where no more than $c$ disks can have a point in common.

These results partially answer some questions of Cabello, Haverkort, van Kreveld, and Speckmann.

\end{abstract}

\section{Introduction}

In cartography and data visualization, one often has to place similar copies of a symbol, typically an opaque disk, on a map or a figure at given locations \cite{De99}, \cite{Gr90}. The size of the symbol is sometimes proportional to the quantitative data associated with the location. On a cluttered map, it is difficult to identify the symbols. Therefore, it has been investigated in several studies how to minimize the amount of overlap \cite{GrC78}, \cite{SlM03}.

\begin{figure}
\centerline{\includegraphics{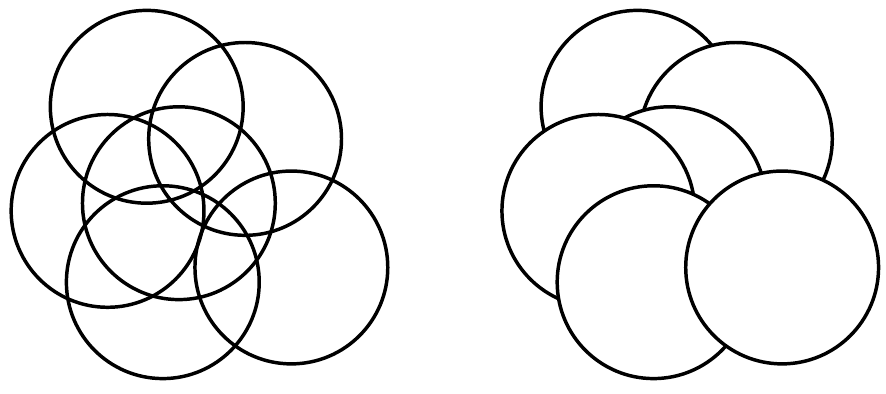}}
\caption{\label{fig_stacking_order}Left: A collection of unit disks in the plane. Right: A stacking order for them.}
\end{figure}

In the present note, we follow the approach of Cabello, Haverkort, van
Kreveld, and Speckmann~\cite{CaH10}. We assume that the symbols used are
opaque circular disks of the same size. Given a collection $\cal D$ of $n$
distinct {\em unit} disks in the $(x,y)$-plane, a {\em stacking order} is a one-to-one assignment $f\; :\;{\cal D} \rightarrow \{1, 2, \ldots, n\}$. We consider the integer $f(D)$ to be the $z$-coordinate of the disk $D\in\cal D$. The {\em map} corresponding to this stacking order is the 2-dimensional view of this arrangement from the point at negative infinity of the $z$-axis (for notational convenience, we look at the arrangement from below rather than from above.) In particular, for the lowest disk $D$, we have $f(D)=1$, and this disk, including its full perimeter, is visible from below. The total length of the
boundary pieces of the disks visible from below is the {\em visible perimeter}
of $\cal D$ with respect to the stacking order $f$, denoted by $\vis({\cal
  D},f)$. We are interested in finding a stacking order for which the visible
perimeter of $\cal D$ is as large as possible. See
Figure~\ref{fig_stacking_order}.

There are other situations in which this setting is relevant. Sometimes the vertices of a graph are not represented as points but as circles of a given radius. It may happen that
some vertices overlap in the visualization (especially if they have
further constraints on their geometric position), and then it becomes important to choose a
convenient stacking order that maximizes the visible perimeter.

Given an integer $n$, we define
\begin{equation}\label{eq_def_v}
v(n) = \inf_{|{\cal D}|=n} \max_{f} \vis({\cal D},f),
\end{equation}
where the maximum is taken over all stacking orders $f$.
We would like to describe the asymptotic behavior of $v(n)$, as $n$ tends to infinity.

Cabello {\em et al.}\ have already noted that $v(n)=\Omega(n^{1/2})$; in other
words, every set $\cal D$ of $n$ disks of unit radii admits
a stacking order with respect to which its visible perimeter is
$\Omega(n^{1/2})$. Indeed, by a well-known result or Erd\H os and
Szekeres~\cite{ErSz35}, we can select a sequence of $\lceil n^{1/2}\rceil$
disks $D_i\in {\cal D}\; (1\le i\le \lceil n^{1/2}\rceil)$ such that their
centers form a monotone sequence. More precisely, letting $x_i$ and $y_i$
denote the coordinates of the center of $D_i$, we have $x_1\le x_2\le
x_3\le\ldots$ and either $y_1\le y_2\le y_3\le\ldots$ or $y_1\ge y_2\ge
y_3\ge\ldots$. Then, in any stacking order $f$ such that $f(D_i)=i$ for every
$i,\; 1\le i\le \lceil n^{1/2}\rceil$, a full quarter of the perimeter of each $D_i\; (1\le i\le \lceil n^{1/2}\rceil)$ is visible from below. Therefore, the visible perimeter of $\cal D$ with respect to $f$ satisfies $$\vis({\cal D},f)\ge \frac{\pi}{2}\lceil n^{1/2}\rceil.$$

At the problem session of {\em EuroCG'11} (Morschach, Switzerland), Cabello, Haverkort, van Kreveld, and Speckmann asked whether $v(n)=\Omega(n)$; in other words, does there exist a positive constant $c$ such that every set of $n$ unit disks in the plane admits a stacking order, with respect to which its visible perimeter is at least $cn$?  We answer this question in the negative; {\em cf.}\ Theorems~\ref{theorem2} and~\ref{theorem5} below.

Given a set of points $P$ in the plane, let ${\cal D}(P)$ denote the
collection of disks of radius $1$ centered at the elements of $P$. For any
positive real $\eps$, let $\eps P$ stand for a similar copy of $P$, scaled by
a factor of $\eps$. For a stacking order $f$ of $\mathcal D(P)$ we will study
the quantity $\vis({\cal D}(\eps P),f)$. (Note the slight abuse of notation:
We denote the stacking order of $\mathcal D(P)$ and the corresponding
stacking order of $\mathcal D(\eps P)$ by the same symbol $f$. The two orders are also identified in Lemmas~\ref{lemma1} and~\ref{lemma2.1} and in Theorems~\ref{theorem2}, \ref{theorem3}, and~\ref{theorem5}.) It is not hard to verify that, as $\eps$ gets smaller, the function $\vis({\cal D}(\eps P),f)$ decreases. To see this, it is enough to observe, as was also done by
Cabello {\em et al.}\ (unpublished), that as we contract the set of centers,
the part of the boundary of each unit disk visible from below shrinks. As we
will see in Lemma~\ref{lemma2.1}, the limit in the following lemma has a simple alternative geometric interpretation.

\begin{lem}\label{lemma1}  For every point set $P$ in the plane and for
  every stacking order $f$ of the collection of disks ${\cal D}(P)$, we have
$$\vis({\cal D}(P),f)\ge\lim_{\eps\rightarrow 0}\vis({\cal D}(\eps P),f).$$
\end{lem}

As in \cite{AlKP89}, \cite{Va92}, and \cite{Va96}, we consider {\em $C$-dense}
$n$-element point sets $P$, {\em i.e}., point sets in which the ratio of the maximum distance between two points to the minimum distance satisfies
$$\frac{\max(|pq| : p,q\in P)}{\min(|pq| : p,q\in P, p\not=q)}\le C n^{1/2}.$$
(The above ratio is sometimes called the \emph{spread} of $P$ \cite{Er03}; thus, we consider point sets with spread at most $Cn^{1/2}$.)

\begin{thm}\label{theorem2} For any $C$-dense $n$-element point set $P$ in the plane and for any stacking order $f$, we have
$$\lim_{\eps\rightarrow 0}\vis({\cal D}(\eps P),f)\le C'n^{3/4},$$
where $C'$ is a constant depending only on $C$.
\end{thm}

The order of magnitude of the upper bound in Theorem~\ref{theorem2} cannot be
improved:

\begin{thm}\label{theorem3} For every positive integer $n$, there exists a $4$-dense $n$-element point set $P_n$ in the plane and a stacking order $f$ such that
$$\lim_{\eps\rightarrow 0}\vis({\cal D}(\eps P_n),f)\ge n^{3/4}.$$
\end{thm}

In the general case, where $P$ is an arbitrary $n$-element point set in the plane, we have been unable to improve on the easy lower bound
$$\max_f \vis({\cal D}(P),f)=\Omega(n^{1/2}),$$
sketched above. However, under special assumptions on $P$, we can do better.

\begin{thm}\label{theorem4} Every $C$-dense $n$-element point set $P$ in the plane admits a stacking order $f$ with
$$ \vis({\cal D}(P),f)\ge C''n^{2/3},$$
where $C''>0$ depends only on $C$.
\end{thm}

In particular, Theorem~\ref{theorem4} provides an $\Omega(n^{2/3})$ lower bound for the
visible perimeter of a collection of $n$ unit disks centered at the points of
an $n^{1/2}\times n^{1/2}$ uniform grid, under
a suitable stacking order. If the side length of the grid is very small, this is
better than the line-by-line ``lexicographic'' stacking order, for which the
visible perimeter is only $\Theta(n^{1/2}\log n)$. It turns out that
in this case there is no stacking order for which the order of the magnitude
of the visible perimeter would exceed $n^{2/3}$.

\begin{thm}\label{theorem5}
Let $n$ be a perfect square and let $G_n$ denote an $n^{1/2}$ by $n^{1/2}$ uniform grid in the plane. For any stacking order $f$, we have
$$\lim_{\eps\rightarrow 0}\vis({\cal D}(\eps G_n),f)=O(n^{2/3}).$$
\end{thm}

Consequently, we have $v(n)=O(n^{2/3})$.

Lemma~\ref{lemma1} implies that the worst collections of disks are those whose centers are very close to each other, so all disks have a point in common. This is, of course, not a
realistic assumption in the labeling problem in cartography that has motivated
our investigations. In practical applications, only a bounded number of unit
disks share a point. For such a case, we have the following result:

\begin{thm}\label{thm_bounded_overlap}
Let $\cal D$ be a collection of $n$ unit disks in which at most $c$ disks have a point in common. Then there exists a stacking order $f$ for which
$$\vis({\cal D}, f) = \Omega(v(c) n / c),$$
where $v(c)$ is given in (\ref{eq_def_v}). This bound is worst-case asymptotically tight.
\end{thm}

In Section 2, we establish Theorems~\ref{theorem2} and~\ref{theorem3}. The proof of Theorem~\ref{theorem4} is
presented in Section 3. In Section~4, we consider the square grid and present a
much simpler proof of this special case of Theorem~\ref{theorem4}
based on Jarnik's theorem~\cite{Ja25}; we then prove Theorem~\ref{theorem5}, which states that the
bound of Theorem~\ref{theorem4} is tight in this case. In Section~5, we prove Theorem~\ref{thm_bounded_overlap}. The last section contains concluding
remarks and open problems.

\section{Dense Sets with Largest Visible Perimeter}

In this section, we prove Theorems~\ref{theorem2} and~\ref{theorem3}.

First, we express the limit of visible perimeters in a simpler form. Given a set of points $P$ in the plane, let $\conv P$
stand for its convex hull. Let $D(p)$ denote the unit disk centered at
$p$ and let ${\cal D}(P)$ stand for the set $\{D(p) : p\in P\}$.

Fix an orthogonal system of coordinates in the plane. For any point $p=(x,y)$ and for any $\eps>0$, let $\eps p$ denote the point with coordinates $(\eps x,\eps y)$.

\begin{lem}\label{lemma2.1} Let $P=\{p_1, p_2,\ldots, p_n\}$ be a set of
  points in the plane, let $\eps>0$, and let $f$ be the stacking order of ${\cal D}(\eps P)$ given by $f(D(\eps p_i))=i$ for $i=1, 2, \ldots, n$.

We have
$$\lim_{\eps\rightarrow 0}\vis({\cal D}(\eps P),f)=\sum_{i=1}^n \tau_i,$$
where $\tau_1=2\pi$, and for all other indices, $\tau_i=0$ if $p_i$ belongs to $\conv\{p_1,\allowbreak p_2,\allowbreak \ldots,\allowbreak p_{i-1}\}$, and $\tau_i$ is equal to the external angle of the convex polygon $\conv\{p_1, p_2,\ldots, p_i\}$ at vertex $p_i$, otherwise.
\end{lem}

\begin{figure}
\centerline{\includegraphics{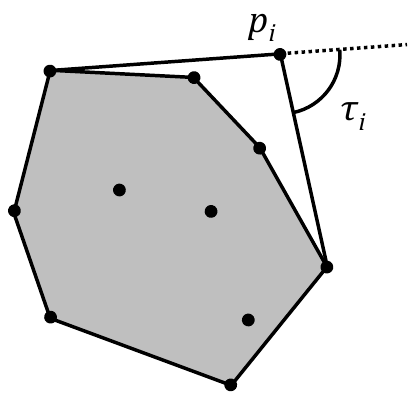}}
\caption{\label{fig_xang}If $p_i$ lies outside the convex hull of the preceding points, then $\tau_i$ is defined as the external angle of the polygon $\conv\{p_1, \ldots, p_i\}$ at vertex $p_i$.}
\end{figure}

See Figure~\ref{fig_xang}.

\begin{proof}[Proof of Lemma~\ref{lemma2.1}] We prove that the contribution of
$\mathcal D(\eps p_i)$ to the visible perimeter tends to $\tau_i$ as
$\eps\to0$ for each $1\le i\le n$.

Since $D(\eps p_1)$ is the lowest disk in ${\cal D}(\eps P)$, its whole
boundary is visible from below. Therefore, its contribution is $2\pi$. Let
$i>1$.
If $p_i$ belongs to the interior of $\conv\{p_1,p_2,\ldots,p_{i-1}\}$, then there is a threshold $\eps_0>0$ such that
$$D(\eps p_i)\subset\bigcup_{j=1}^{i-1}D(\eps p_j),$$
for every $\eps<\eps_0$. In this case, no portion of the boundary of $D(\eps
p_i)$ is visible from below, provided that $\eps$ is sufficiently small. If
$p_i$ lies on the boundary of $\conv\{p_1, p_2,\ldots, p_i\}$, then it is
in between some points $p_j$ and $p_k$ with $1\le j<k<i$ and although $\mathcal
D(\eps p_i)$ will not be entirely covered by earlier disks for any $\eps>0$,
the part of its boundary outside $\mathcal D(\eps p_j)\cup\mathcal D(\eps
p_k)$ tends to zero as $\eps\to 0$.

Finally, if $p_i$ lies outside $\conv\{p_1,\ldots,p_{i-1}\}$, then it is a
vertex of $\conv\{p_1,\allowbreak \ldots,\allowbreak p_i\}$. Consider the external unit normal vectors
to the two sides of $\conv\{p_1,\ldots p_i\}$ that meet at $\eps p_i$ (or in
case the convex hull is a single segment, the two unit normal vectors
for this segment). Drawing these vectors from $\eps p_i$, the arc on the
boundary of $\mathcal D(p_i)$ between them is of length $\tau_i$ and it is not covered by $\bigcup_{j=1}^{i-1}D(\eps p_j)$. Thus, it is visible from below, and, as $\eps\rightarrow 0$, the total contribution of the remaining part of the boundary of $D(\eps p_i)$ to the visible perimeter tends to $0$, concluding the proof.
\end{proof}

{
\renewcommand{\thethm}{2}
\begin{thm}For any $C$-dense $n$-element point set $P$ in the plane and for any stacking order $f$, we have
$$\lim_{\eps\rightarrow 0}\vis({\cal D}(\eps P),f)\le C'n^{3/4},$$
where $C'$ is a constant depending only on $C$.
\end{thm}
\addtocounter{thm}{-1}
}

\begin{proof}
Consider a $C$-dense point set $P$ in the
plane and let $f$ be a stacking order for $\mathcal D(P)$. Using Lemma~\ref{lemma2.1}, it
is enough to prove $\sum_{i=1}^n\tau_i\le C'n^{3/4}$ for the angles $\tau_i$ defined
in the lemma. As $\tau_i=0$ whenever $p_i$ is contained in
$\conv\{p_1,\ldots,p_{i-1}\}$, we can assume this is never the case.

Since the quantity $\sum \tau_i$ is independent of scale, we can assume without
loss of generality that the minimum distance between points is
$1$; thus, the maximum distance (diameter) is at most $C n^{1/2}$. We write
$P=\{p_1, p_2,\ldots, p_n\}$ with $f(D(p_i))=i$.

\begin{figure}
\centerline{\includegraphics{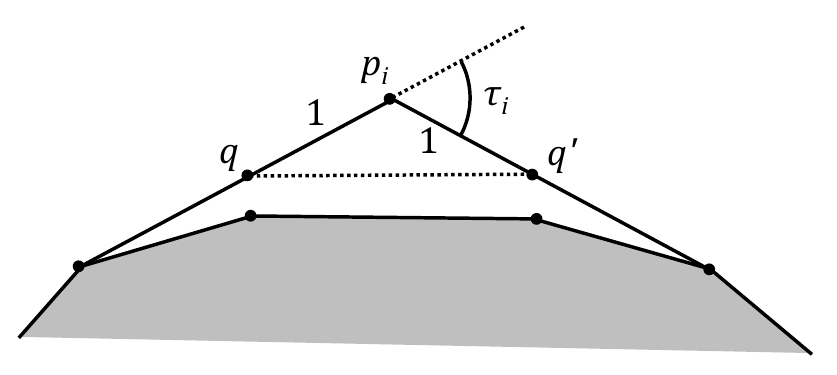}}
\caption{\label{fig_perimeter}The triangle $p_iqq'$ lies entirely outside the convex hull of $p_1, \ldots, p_{i-1}$.}
\end{figure}

For every $i\; (1\le i\le n)$, let $\per(i)$ denote the perimeter of
$\conv\{p_1, p_2,\ldots, p_i\}$. We define the perimeter of a segment to be
twice its length and the perimeter of a point to be $0$. Let
$2\le i\le n$, consider the two sides of the polygon $\conv\{p_1,
p_2,\ldots, p_i\}$ meeting at $p_i$, and denote by $q$ and $q'$ the points on these sides
at unit distance from $p_i$. Since no point of $P$ is closer to $p_i$ than
$1$, the triangle $p_iqq'$ does not contain any element of $\{p_1, p_2,\ldots,
p_{i-1}\}$. (See Figure~\ref{fig_perimeter}.) Hence, $\conv\{p_1, p_2,\ldots, p_{i-1}\}$ is contained in the
convex region obtained from $\conv\{p_1,\allowbreak p_2,\allowbreak \ldots,\allowbreak p_i\}$ by cutting off the
triangle $p_iqq'$. (In the degenerate case when $\conv\{p_1,\allowbreak \ldots,\allowbreak p_i\}$ is a
segment, we have $q=q'$,
and the empty ``triangle'' becomes just a unit segment.) This observation implies that the perimeter of $\conv\{p_1, p_2,\ldots, p_{i-1}\}$ satisfies
$$\per(i-1)\le \per(i)-|p_iq|-|p_iq'|+|qq'|= \per(i)-2+2\cos\frac{\tau_i}{2}\le \per(i)-\frac{\tau_i^2}{5}.$$
Here we used that the external angle of the triangle $p_iqq'$ at vertex $p_i$ is $\tau_i$.

Thus, we have $$\per(i)-\per(i-1)\ge \frac{\tau_i^2}{5},$$
for all $i>1$. Adding up these inequalities, we obtain
$$\per(n)\ge \sum_{i=2}^n\frac{\tau_i^2}{5}.$$
Since $\per(n)$ is at most $\pi$ times the diameter of $P$, that is, $\per(n)\le\pi C n^{1/2}$, we have
$$\sum_{i=2}^n \tau_i^2\le 5\pi C n^{1/2}.$$
Applying the relationship between the arithmetic and quadratic means, we can conclude that
$$\sum_{i=2}^n \tau_i\le (n-1)^{1/2}\left(\sum_{i=2}^n\tau_i^2\right)^{1/2}<(5\pi C)^{1/2}n^{3/4}.$$

Taking into account that $\tau_1=2\pi$, the theorem follows by Lemma~\ref{lemma2.1}.
\end{proof}

{
\renewcommand{\thethm}{3}
\begin{thm}For every positive integer $n$, there exists a $4$-dense $n$-element point set $P_n$ in the plane and a stacking order $f$ such that
$$\lim_{\eps\rightarrow 0}\vis({\cal D}(\eps P_n),f)\ge n^{3/4}.$$
\end{thm}
\addtocounter{thm}{-1}
}

\begin{figure}
\centerline{\includegraphics{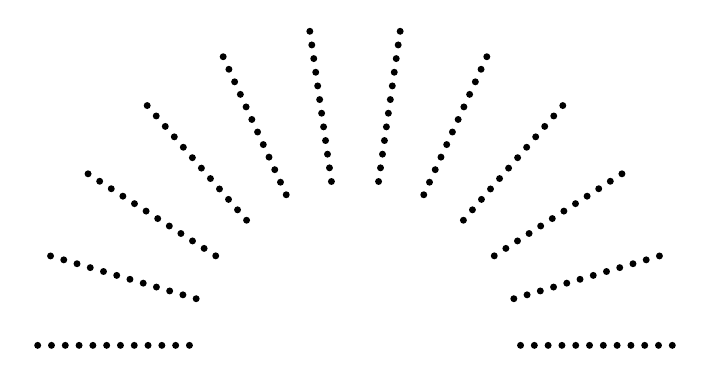}}
\caption{\label{fig_concentric}A dense point set that has a good stacking order.}
\end{figure}

\begin{proof}
Suppose for simplicity that $n=k^2$ for some integer $k\ge3$. Our point set $P_n$ consists of the points having polar coordinates
$(r, \theta) = (i, j\pi/(k-1))$ for $i\in\{k,k+1,\ldots,2k-1\}$ and $j\in\{0,1,\ldots,k-1\}$. See Figure~\ref{fig_concentric}. The smallest
distance between two points in $P_n$ is $1$, and the largest distance
is less than $4k$; thus, $P_n$ is $4$-dense, as required.

Our stacking order $f$ takes the points by increasing $r$, and for each $r$ by increasing~$\theta$.

We apply Lemma~\ref{lemma2.1} and calculate the sum of
the external angles determined by $f$. Denote by $C_i$ the circle of radius $i$ centered at the origin.
Consider a point $p\in P_n$ on $C_i$. Let $\ell$ be the ray leaving $p$ towards the right tangent to $C_i$, and let $\ell'$ be the ray leaving $p$ towards the left tangent to $C_{i-1}$. Let $q$ be the point of tangency between $\ell'$ and $C_{i-1}$. Then all the points of $P_n$ that precede $p$ in the order $f$ lie below $\ell$ and $\ell'$. Thus, the external angle $\tau$ contributed by $p$ is at
least the supplement $\alpha$ of the angle between $\ell$ and $\ell'$. We have 
$\alpha=\measuredangle p0q\ge\sin\alpha=\sqrt{2i-1}/i\ge n^{-1/4}$.
The theorem follows.
\end{proof}

\section{All Dense Sets Have Good Stacking Orders}

We now turn to Theorem~\ref{theorem4}.

{
\renewcommand{\thethm}{4}
\begin{thm} Every $C$-dense $n$-element point set $P$ in the plane admits a stacking order $f$ with
$$ \vis({\cal D}(P),f)\ge C''n^{2/3},$$
where $C''>0$ depends only on $C$.
\end{thm}
\addtocounter{thm}{-1}
}

Throughout this section, let $P$ be a $C$-dense $n$-point set
in the plane. We will define a stacking order $f$ for $\mathcal D(P)$ for
which the
external angles $\tau_i$ defined in Lemma~\ref{lemma2.1} satisfy
$\sum_{i=1}^n \tau_i \ge C''n^{2/3}$, for some constant
$C''>0$ depending only on $C$. Then the theorem follows from
Lemma~\ref{lemma2.1}.

Assume without loss of generality that the minimum  distance in
$P$ is $1$. Then, since $P$ is $C$-dense, there exists a disk
of radius $C n^{1/2}$ that contains all of $P$. Let $D$ be such
a disk, and let $K$ be a circle of radius $2Cn^{1/2}$
concentric with $D$.

\begin{figure}
\centerline{\includegraphics{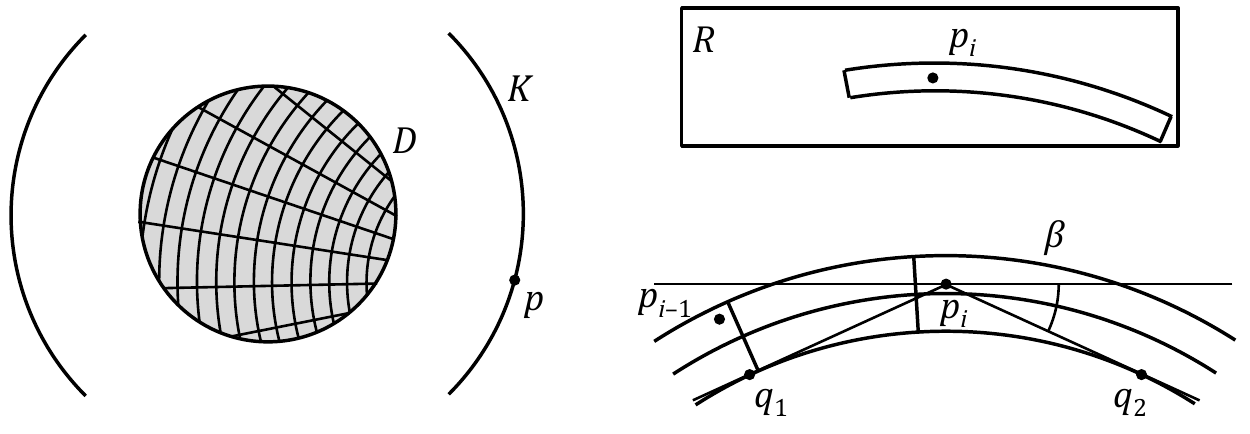}}
\caption{\label{fig_annular_sectors}Left: Partition of $D$ into annular sectors
centered at a point $p\in K$. Top right: The sector containing $p_i$ is contained in
the rectangle $R$ centered at $p_i$. Bottom right: Point $p_i$ contributes
external angle at least $\beta$.}
\end{figure}

Given a point $p\in K$, we define a family $F = F(p)$  of
annular sectors that disjointly cover the plane, as follows:
For each positive integer $i$, let $K_i = K_i(p)$ be a circle
centered at $p$ with radius $i n^{-1/6}$; then divide each
annulus between two consecutive circles into sectors of angular
length $\alpha = C^* n^{-1/3}$ for a large enough constant
$C^*$ (as will be specified below). See
Figure~\ref{fig_annular_sectors} (left).

Note that each annular sector that intersects $D$ has area
$\Theta(1)$ (since the radius of such a sector is $\Theta(n^{1/2})$).
The number of annular sectors that intersect $D$ is $\Theta(n^{1/2} n^{1/6} n^{1/3}) =\Theta(n)$.
Call a sector \emph{occupied} if it contains at
least one point of $P$.

\begin{lem}\label{lemma3.1} There exists a point $p\in K$
for which $\Omega(n)$ sectors of $F(p)$ are occupied.
\end{lem}

\begin{proof} Choose $p$ uniformly at random on $K$ and construct the
sectors using $p$ and dividing the annuli into the correct-length sectors in
an arbitrary way.
For each point $p_i\in P$, define the random variable $n(p_i)$
to be the number of points of $P$ contained in the sector of
$F(p)$ that contains $p_i$. We claim that the expected value
$E[n(p_i)]$ of $n(p_i)$ satisfies
\begin{equation*}
E[n(p_i)] \le k
\end{equation*}
for some constant $k$.

Indeed, let $R = R_{p_i}(p)$ be a rectangle centered at $p_i$,
with dimensions $(k' n^{1/6})\times(k' n^{-1/6})$, and with
short sides parallel to the line $pp_i$, for an appropriate
constant $k'$. If $k'$ is large enough (but constant with
respect to $n$), then $R$ completely contains the sector of
$F(p)$ that contains $p_i$. See Figure~\ref{fig_annular_sectors} (top right).
Thus, it suffices to bound the
expected number of points of $P$ in $R$. Note that, as $p$
rotates around $K$, $R$ rotates around its center together with
$p$.

Partition the plane into annuli centered at $p_i$ by tracing
circles around $p_i$ of radii $1,2,4,8,\ldots$. The annulus
with inner radius $r$ and outer radius $2r$ contains at most
$k_2 r^2$ points of $P$, for some constant $k_2$. Each such
point has probability at most $k_3 n^{-1/6}r^{-1}$ of falling
in $R$ (over the choice of $p$), for another constant $k_3$;
therefore, the expected contribution of this annulus to
$n(p_i)$ is at most $k_2k_3 r n^{-1/6}$. Summing up for all
annuli with inner radius $r \le k' n^{1/6}$, we obtain that
$E[n(p_i)] \le k$ for some constant $k$, as claimed.

Now, call point $p_i$ \emph{isolated} if $n(p_i) \le 2k$. By
Markov's inequality, each point $p_i$ has probability at least
$1/2$ of being isolated. Therefore, the expected number of
isolated points is at least $n/2$. There must exist a $p$ that
achieves this expectation, and for it we obtain at least
$n/(4k)$ occupied sectors, proving the lemma.
\end{proof}

\begin{proof}[Proof of Theorem~\ref{theorem4}] Fix a point $p$ for which
$F(p)$ has $\Omega(n)$ occupied sectors. Color the sectors with
four colors, using colors $1$ and $2$ alternatingly on the
odd-numbered annuli and colors $3$ and $4$ alternatingly on the
even-numbered annuli.

There must be a color for which $\Omega(n)$ sectors are
occupied. Consider only the occupied sectors with this color.
Let these sectors be $S_1, S_2, \ldots, S_m$, listed by
increasing distance from $p$, and for each fixed distance, in
clockwise order around $p$. Select one point $p_i\in P\cap S_i$
from each of these sectors. Let the stacking order $f$ start with these
points, that is, $f(\mathcal D(p_i))=i$ for $i=1,\ldots,m$. The order of the
remaining points in $P$ is arbitrary.

We claim that each selected point $p_i$ contributes an external
angle of $\tau_i = \Omega(n^{-1/3})$, which implies that $\sum
\tau_i = \Omega(n^{2/3})$, as desired.

Indeed, consider the $i$-th selected point $p_i$. Suppose
without loss of generality that $p$ lies directly below $p_i$.
Let $K_k$ and $K_{k+1}$ be the inner and outer circles bounding
the annulus that contains $p_i$. Trace rays $z_1$ and $z_2$
from $p_i$ tangent to $K_{k-1}$, touching $K_{k-1}$ at points
$q_1$ and $q_2$. See Figure~\ref{fig_annular_sectors} (bottom right).

Every point $p_j$, $j<i$, that is \emph{not} contained in the
same annulus as $p_i$ lies below these rays. Moreover, the
angle $\beta$ that these rays make with the horizontal is
$\Theta(n^{-1/3})$: Consider, for example, the ray $z_1$. The
triangle $pp_iq_1$ is right-angled, with angle $\measuredangle
p_ipq_1 = \beta$. We have $pq_1 = \Theta(n^{1/2})$ and $pp_i =
pq_1 + \Theta(n^{-1/6})$. It follows that $p_iq_1 =
\Theta(n^{1/6})$, and so $\beta \approx \tan \beta = p_iq_1 /
pq_1 = \Theta(n^{-1/3})$.

Now suppose that $p_{i-1}$ lies in the same annulus as $p_i$.
If the constant $C^*$ in the definition of $\alpha$ is chosen
large enough, then $p_{i-1}$ must have a smaller $y$-coordinate
than $p_i$. (In the worst case, $p_i$ lies near the bottom-left
corner of its sector and $p_{i-1}$ lies near the top-right
corner of its sector.)

Thus, $p_i$ contributes external angle $\tau_i\ge \beta =
\Omega(n^{-1/3})$, as claimed.
\end{proof}

\section{The ``Worst'' Dense Set: the Grid}

In this section, we assume that $n$ is a square number and $G_n$ denotes an $n^{1/2}$ by $n^{1/2}$ integer grid. Note that $G_n$ is a $\sqrt2$-dense set consisting of $n$ points.

As we mentioned in the Introduction, in the special case where $P=\eps G_n$, Theorem~\ref{theorem4} has a simple proof. For $\mathcal D(\eps G_n)$, one can produce a stacking order with large visible perimeter using the following greedy algorithm (which can also be applied to any other point set $P$): Set $P_n=G_n$, and select a vertex of $\conv(P_n)$ whose external angle is maximum. Let this vertex be $p_n$, the last element in the desired order $f_{\text{greedy}}$. Repeat the same step for the set $P_{n-1} = P_n \setminus \{p_n\}$, and continue in this fashion until the first element $p_1$ gets defined.

By Jarnik's theorem~\cite{Ja25}, every convex polygon has $O(n^{1/3})$ vertices in $G_n$. Therefore, at each step, the greedy algorithm selects a point $p_i$ that makes an external angle $\tau_i = \Omega(n^{-1/3})$. Hence, $\sum\tau_i = \Omega(n^{2/3})$ for the order $f_{\text{greedy}}$. Lemma~\ref{lemma2.1} completes the proof.

Now we turn to Theorem~\ref{theorem5}.

{
\renewcommand{\thethm}{5}
\begin{thm}
Let $n$ be a perfect square and let $G_n$ denote an $n^{1/2}$ by $n^{1/2}$ uniform grid in the plane. For any stacking order $f$, we have
$$\lim_{\eps\rightarrow 0}\vis({\cal D}(\eps G_n),f)=O(n^{2/3}).$$
\end{thm}
\addtocounter{thm}{-1}
}

Our proof is an improved version of the
proof of Theorem~\ref{theorem2}. There we were concerned with how the {\em perimeter} of the
convex hull grows as we add the points of our set one by one as prescribed by
the stacking order. As is well known, the perimeter of a convex set in the plane is the
integral of its width in all directions (this is known as Cauchy's theorem; see {\it e.g.} \cite{PaA95}, Theorem 16.15).
The proof of Theorem~\ref{theorem5} is very
similar, but we deal with the widths in different directions
in a non-uniform way. The width in a direction close to the
direction of a short grid vector is more important in the analysis than
widths in other directions.

\begin{proof}[Proof of Theorem~\ref{theorem5}.]
Let $G_n=\{p_1,\ldots,p_n\}$ be an enumeration of the points of $G_n$ according to a given stacking order, and let $\tau_i$ denote the corresponding external angles, as defined in Lemma~\ref{lemma2.1}. According to the lemma, we need to prove that $\sum_{i=1}^n\tau_i=O(n^{2/3})$. Let us partition this sum into several parts, and bound the contribution of each part separately.

Let $[n]=\{1,\ldots,n\}$. We start with the small angles. Let
$$I_0=\{i\in[n]\mid\tau_i<n^{-1/3}\}.$$ Clearly, we have $$\sum_{i\in
I_0}\tau_i<n\cdot n^{-1/3}=n^{2/3}.$$

\begin{figure}
\centerline{\includegraphics{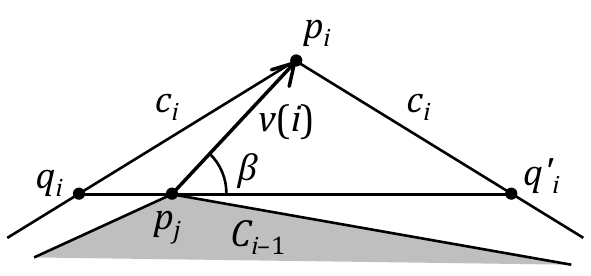}}
\caption{\label{fig_v_beta}The triangle $p_iq_iq'_i$ is the largest isosceles triangle at point $p_i$ that does not intersect the interior of $C_{i-1}$.}
\end{figure}

As in the proof of Theorem~\ref{theorem2}, let $C_i=\conv\{p_1,\ldots,p_i\}$ and denote the perimeter of $C_i$ by $\per(i)$. Since $G_n$ is an $n^{1/2}\times n^{1/2}$ integer grid, we have $\per(n)=4(n^{1/2}-1)$. Consider only those indices $i>1$ that do not belong to $I_0$. For these indices, we have $\tau_i>0$, so that $p_i$ must be a vertex of $C_i$. For each such point $p_i$, let $c_i$ denote the smallest number satisfying the following condition: the segment connecting the points $q_i$ and $q_i'$ that lie on the boundary of $C_i$ at distance $c_i$ from $p_i$, intersects $C_{i-1}$. (In the case where $C_i$ is a segment, we have $q_i=q_i'\in C_{i-1}$.) Note that the segment $q_iq_i'$ contains a point $p_j$ with $1\le j<i$.
See Figure~\ref{fig_v_beta}.

In the proof of Theorem~\ref{theorem2}, we argued that $\per(i)-\per(i-1)\ge\tau_i^2/5$. Now the same
argument gives that $\per(i)-\per(i-1)>c_i\tau_i^2/5$. Let
$$I_1=\{i\in[n]\setminus(I_0\cup\{1\})\mid c_i\tau_i>n^{-1/6}\}.$$ For $i\in
I_1$, we have $\per(i)-\per(i-1)\ge\tau_in^{-1/6}/5$. Since $\per(i)$ is
monotone in $i$, we conclude that $$\sum_{i\in
I_1}\tau_i\le5n^{1/6}(\per(n)-\per(1))<20n^{2/3}.$$

Let $$I_2=[n]\setminus(\{1\}\cup I_0\cup I_1).$$ To bound the angles $\tau_i$
for indices $i\in I_2$, we need a charging scheme and we need to consider the
growth of the width of $C_i$ in some specific directions. The {\em width} of a
planar set in a given direction is the diameter of the orthogonal projection
of the set to a line in this direction. Let us associate the directions in the
plane with the points of the unit circle $K$. We identify opposite points of
this circle as the widths of the same set in opposite
directions are the same. This makes the total length of $K$ become $\pi$. We
define a set of arcs along $K$ as follows.
For any non-zero grid vector $v$ from the integer grid and for any integer
$\ell\ge0$, let $V_{v,\ell}$ denote the arc of length $2^{-\ell}$ symmetric
around the direction of $v$. For any direction $\alpha\in K$, let
$\rho_i(\alpha)$ denote the width of $C_i$ in the direction {\em orthogonal} to
$\alpha$ ({\em i.e.}, where the corresponding projection is parallel to $\alpha$).

The perimeter $\per(i)$ is equal to the integral of $\rho_i(\alpha)$ along the circle $K$ (note that after the identification of opposite points the length of $K$ became $\pi$). We have $\rho_i(\alpha)=\rho_{i-1}(\alpha)$, unless the direction $\alpha$ is tangent to $C_i$ at the vertex $p_i$. Let $U_i$ denote the arc of directions where such a tangency occurs. Clearly, the length of $U_i$ is $\tau_i$, and for any arc $V$ that contains $U_i$, we have
$$\int_{V}(\rho_i(\alpha)-\rho_{i-1}(\alpha))d\alpha=\per(i)-\per(i-1)\ge
c_i\tau_i^2/5.$$

For each index $i\in I_2$, choose a grid point $p_j$ on the segment $q_iq'_i$. (Recall that the points $q_i$ and $q'_i$ are at distance $c_i$ from $p_i$, and that there is always a grid point between them.) We {\em charge} the index $i$ to the pair $(v(i),\ell(i))$, where $v(i)$ is the grid vector pointing from $p_j$ to $p_i$ and $\ell(i)$ is the largest integer such that $V_{v(i),\ell(i)}$ contains $U_i$. Notice that $|v(i)|\le c_i$. Denote by $I_2(v,\ell)$ the set of indices $i\in I_2$ that are charged to the pair $(v,\ell)$.

Note that $U_i$ is symmetric around the direction of the segment
$q_iq'_i$. For the angle $\beta$ between this direction and the
direction of $v(i)$ we have
$|v(i)|\sin\beta=c_i\sin(\tau_i/2)$ (refer again to Figure~\ref{fig_v_beta}). This implies
$\beta<c_i\tau_i/|v(i)|$, and hence
$2^{-\ell(i)}<4\beta+2\tau_i<6c_i\tau_i/|v(i)|$. Finally, we also have
$$\int_{V_{v(i),\ell(i)}}(\rho_i(\alpha)-\rho_{i-1}(\alpha))d\alpha\ge
c_i\tau_i^2/5>2^{-\ell(i)}|v(i)|\tau_i/30.$$

Let $s(v,\ell)=\sum_{i\in I_2(v,\ell)} \tau_i$. The integral $\int_{V_{v,\ell}}\rho_i(\alpha)d\alpha$
is monotone in $i$ and grows by at least $2^{-\ell}|v|\tau_i/30$ at every
$i \in I_2(v,\ell)$. We have $\rho_1(\alpha)=0$ and
$\rho_n(\alpha)<(2n)^{1/2}$, so that the final integral satisfies
$\int_{V_{v,\ell}}\rho_n(\alpha)d\alpha \le 2^{-\ell}(2n)^{1/2}$. Therefore, $\sum_{i\in I_2(v,\ell)} \bigl( 2^{-\ell}|v|\tau_i/30 \bigr) \le 2^{-\ell}(2n)^{1/2}$, which implies that $s(v,\ell)\le 30\sqrt{2}n^{1/2}/|v|$.

Consider the set of all pairs $(v,\ell)$ such that there is an index $i\in I_2$ charged to them. We have $c_i\tau_i\le
n^{-1/6}$, $\tau_i\ge n^{-1/3}$ and $|v|\le c_i$, which implies that $|v|\le
n^{1/6}$. We proved that $2^{-\ell}<6c_i\tau_i/|v|\le6n^{-1/6}/|v|$.
On the other hand, we also have $2^{-\ell}\ge2\tau_i\ge2n^{-1/3}$. Thus,
for any given grid vector $v$, there are at most $\log(6n^{1/6}/|v|)$
possible values of $\ell$, where $\log$ denotes the binary logarithm.

Hence, $$\sum_{i\in I_2}\tau_i=\sum_{v,\ell}s(v,\ell)\le\sum_{|v|\le
n^{1/6}}\frac{30\sqrt2n^{1/2}}{|v|}\log\frac{6n^{1/6}}{|v|}.$$
To evaluate this sum, we note that the number of grid vectors $v$ satisfying $2^k\le |v| < 2^{k+1}$ is $\Theta(2^{2k})$. Thus,
$$\sum_{i\in I_2}\tau_i = O{\left( n^{1/2}\sum_{k=0}^{\log n^{1/6}} (\log n^{1/6} - k) 2^k \right)} = O(n^{1/2}n^{1/6}) = O(n^{2/3}).$$
In conclusion, we have $$\sum_{i=1}^n\tau_i=\tau_1+\sum_{i\in I_0}\tau_i+\sum_{i\in
I_1}\tau_i+\sum_{i\in
I_2}\tau_i=2\pi+O(n^{2/3})+O(n^{2/3})+O(n^{2/3})=O(n^{2/3}),$$
completing the proof of the theorem.
\end{proof}

\section{Collections of disks with bounded overlap}

In this section, we prove Theorem~\ref{thm_bounded_overlap}.

{
\renewcommand{\thethm}{6}
\begin{thm}
Let $\cal D$ be a collection of $n$ unit disks in which at most $c$ disks have a point in common. Then there exists a stacking order $f$ for which
$$\vis({\cal D}, f) = \Omega(v(c) n / c),$$
where $v(c)$ is given in (\ref{eq_def_v}). This bound is worst-case asymptotically tight.
\end{thm}
\addtocounter{thm}{-1}
}

Note that Lemma~\ref{lemma2.1} is not relevant in this case, since we cannot contract the set of centers of ${\cal D}$.

\begin{proof}
Partition the plane into an infinite grid of axis-parallel square cells of side-length $4$, where the position of the grid is chosen uniformly at random. For each unit disk, the probability that it belongs entirely to a single cell is $1/4$. Thus, we
can fix the grid in such a way that at least $n/4$ disks lie entirely in a cell. Let $k_i$ be the number of disks entirely contained in cell $i$. By area considerations, we have $k_i \le (16/\pi)c$.

For each cell $i$, we independently select a stacking order that achieves visible perimeter at least $v(k_i)$; then we place all the remaining disks behind them. Thus, our stacking order achieves visible perimeter at least $\sum_i v(k_i)$.

For any $n$-element point set $\cal D$, we can take an $rn$-element point set
$\cal D'$ as the union of $r$ pairwise disjoint translates of $\cal D$. We
clearly have $\max_f\vis({\cal D}',f)\le r\max_f\vis({\cal D},f)$. This implies that
$v(rn) \le r v(n)$. Let $r_i=\lceil c/k_i\rceil\le(16/\pi)c/k_i$, and we have
$v(c)\le v(r_ik_i)\le r_iv(k_i)$, thus
$$v(k_i)\ge v(c)/r_i\ge\frac{\pi v(c)}{16 c}k_i.$$

Since $\sum k_i \ge n/4$, the claimed bound follows.

To show that this bound is worst-case asymptotically tight, take the union of $\lceil
n/c\rceil$ worst-case sets of $c$ disks far from each other.
\end{proof}

\section{Concluding remarks}

\noindent{\bf A.} The greedy algorithm described at the beginning of
Section~4 was first considered by Cabello {\em et al.}\ (unpublished) in the context
of maximizing the \emph{minimum} visible perimeter of a single disk.
They showed that the order $f_{\text{greedy}}$ is always optimal
for this purpose. Unfortunately, this stacking order is \emph{not}
always optimal with respect to the total visible perimeter. Indeed, let $n$ be
a perfect square and consider the set of points $\{p_i\mid1\le i\le n\}$, where
the polar coordinates of $p_i$ are $(r_i, \theta_i) = (e^{bi},2\pi i/n^{1/2})$ with $b>0$
sufficiently small. This point set is obtained as the intersection of
$n^{1/2}$ equally-spaced rays
emanating from the origin, and $n^{1/2}$ ``rounds'' of a very tight
logarithmic spiral centered at the origin. The greedy algorithm produces the
stacking order indicated by the indices, so it takes the
points of $P$ outwards along the spiral. The contribution $\tau_i$ is equal for
every point $p_i$ with $i\ge n^{1/2}$ and tends to $2\pi/n^{1/2}$ as $b$
goes to zero, making
$\sum\tau_i = \Theta(n^{1/2})$ if $b$ is small enough. However, taking the points ray by ray in a cyclic
order, going outwards along each ray, the contribution $\tau_i$ is a constant
for the first half of the points, making $\sum\tau_i = \Theta(n)$.

\medskip

\noindent{\bf B.}
Theorem 4 can be generalized to point sets satisfying weaker density conditions. Indeed, let $P$ be a set of $n$ points in the plane with diameter $D$ and minimum distance $d$. A randomized construction, similar to the one used in the proof Theorem~4, guarantees the existence of a stacking order $f$ such that $\vis(\mathcal
D(P),f)=\Omega(n/(D/d)^{2/3})$. This beats the $\Omega(n^{1/2})$ bound
mentioned in the Introduction as long as $D/d=o(n^{3/4})$.

\medskip

\subsubsection*{Acknowledgements} The authors express their gratitude to Radoslav Fulek and Andres Ruiz Vargas (EPFL), for many insightful discussions on the subject, as well as to the anonymous referees for their useful comments.


\begin{thebibliography}{}

    \bibitem[AlKP89]{AlKP89}
    N. Alon, M. Katchalski, and W. R. Pulleyblank: The maximum size of a convex polygon in a restricted set of points in the plane, {\em Discrete Comput. Geom.} {\bf 4} (1989), 245--251.

	\bibitem[CaH10]{CaH10}
	S. Cabello, H. Haverkort, M. van Kreveld, and B. Speckmann:
    Algorithmic aspects of proportional symbol maps, {\em Algorithmica} {\bf 58} (2010), 543--565.

    \bibitem[De99]{De99}
    B. Dent: {\em Cartography. Thematic Map Design, 5th edn}, McGraw-Hill, New York, 1999.

    \bibitem[ErSz35]{ErSz35}
    P. Erd\H os and G. Szekeres: A combinatorial problem in geometry, {\em Compositio Math.} {\bf 2} (1935), 463--470.
    
    \bibitem[Er03]{Er03} J. Erickson: Nice point sets can have nasty Delaunay triangulations, {\em Discrete Comput. Geom.} {\bf 30} (2003), 109--132.

    \bibitem[Gr90]{Gr90}
    T. Griffin: The importance of visual contrast for graduated circles, {\em Cartography} {\bf 19} (1990), 21--30.

    \bibitem[GrC78]{GrC78}
    R. E. Groop and D. Cole: Overlapping graduated circles: Magnitude estimation and method of portrayal, {\em Can. Cartogr.} {\bf 15} (1978), 114--122.

    \bibitem[Ja25]{Ja25}
    V. Jarn\'{\i}k: \"Uber die Gitterpunkte auf konvexen Kurven, {\em Mathematische Zeitschrift} {\bf 24} (1926), 500--518.

    \bibitem[PaA95]{PaA95}
    J. Pach and P. K. Agarwal: {\em Combinatorial Geometry}, Wiley, New York, 1995.

    \bibitem[SlM03]{SlM03}
    T. A. Slocum, R. B. McMaster, F. C. Kessler, and H. H. Howard: {\em Thematic Cartography and Geographic Visualization, 2nd edn}, Prentice Hall, New York, 2003.

    \bibitem[Va92]{Va92}
    P. Valtr: Convex independent sets and 7-holes in restricted planar point sets, {\em Discrete Comput. Geom.} {\bf 7} (1992), 135--152.

    \bibitem[Va96]{Va96}
    P. Valtr: Lines, line-point incidences and crossing families in dense sets, {\em Combinatorica} {\bf 16} (1996), 269--294.

\end{thebibliography}
\end{document}